\documentclass[11pt]{article} 

\usepackage[utf8]{inputenc} 
\usepackage{amsmath,color,amsthm,amsfonts}
\theoremstyle{plain}
\newtheorem{lem}{Lemma}
\newtheorem{prop}{Proposition}
\newtheorem{thm}{Theorem}

\theoremstyle{definition}
\newtheorem{defin}{Definition}

\usepackage{graphicx} 

\usepackage{booktabs} 
\usepackage{array} 
\usepackage{paralist} 
\usepackage{verbatim} 
\usepackage{subfig} 

\usepackage[nottoc,notlof,notlot]{tocbibind} 
\usepackage[titles,subfigure]{tocloft} 

\newcommand{\exNt}[1]{ ||{#1}||_2}
\newcommand{\exNo}[1]{ ||#1||_1}

\def\showauthornotes{0}

\def\showkeys{0}
\def\showdraftbox{0}
\def\showcolorlinks{1}
\def\usemicrotype{1}
\def\showfixme{0}

\usepackage{etex}

\usepackage[l2tabu, orthodox]{nag}
\usepackage{comment}

\usepackage{xspace,enumerate}

\usepackage[dvipsnames]{xcolor}

\usepackage[T1]{fontenc}
\usepackage[full]{textcomp}

\usepackage[american]{babel}

\usepackage{mathtools}

\usepackage{bm}

\usepackage{amsthm}

\newtheorem{theorem}{Theorem}[section]
\newtheorem*{theorem*}{Theorem}

\newtheorem*{proposition*}{Proposition}
\newtheorem{lemma}[theorem]{Lemma}
\newtheorem*{lemma*}{Lemma}

\newtheorem*{conjecture*}{Conjecture}
\newtheorem{fact}[theorem]{Fact}
\newtheorem*{fact*}{Fact}
\newtheorem{hypothesis}[theorem]{Hypothesis}
\newtheorem*{hypothesis*}{Hypothesis}

\theoremstyle{definition}
\newtheorem{definition}[theorem]{Definition}

\theoremstyle{remark}

\newtheorem*{claim*}{Claim}

\newtheorem*{remark*}{Remark}

\newtheorem*{observation*}{Observation}

\usepackage[letterpaper,
top=1.3in,
bottom=1.3in,
left=1in,
right=1in]{geometry}

\usepackage[varg]{pxfonts} 

\ifnum\showkeys=1
\usepackage[color]{showkeys}
\fi

\ifnum\showcolorlinks=1
\usepackage[
pagebackref,
letterpaper=true,
colorlinks=true,
urlcolor=blue,
linkcolor=blue,
citecolor=OliveGreen,
]{hyperref}
\fi

\ifnum\showcolorlinks=0
\usepackage[
pagebackref,
letterpaper=true,
colorlinks=false,
pdfborder={0 0 0}
]{hyperref}
\fi

\usepackage{prettyref}

\newcommand{\savehyperref}[2]{\texorpdfstring{\hyperref[#1]{#2}}{#2}}

\newrefformat{eq}{\savehyperref{#1}{\textup{(\ref*{#1})}}}
\newrefformat{lem}{\savehyperref{#1}{Lemma~\ref*{#1}}}
\newrefformat{def}{\savehyperref{#1}{Definition~\ref*{#1}}}
\newrefformat{thm}{\savehyperref{#1}{Theorem~\ref*{#1}}}
\newrefformat{cor}{\savehyperref{#1}{Corollary~\ref*{#1}}}
\newrefformat{cha}{\savehyperref{#1}{Chapter~\ref*{#1}}}
\newrefformat{sec}{\savehyperref{#1}{Section~\ref*{#1}}}
\newrefformat{app}{\savehyperref{#1}{Appendix~\ref*{#1}}}
\newrefformat{tab}{\savehyperref{#1}{Table~\ref*{#1}}}
\newrefformat{fig}{\savehyperref{#1}{Figure~\ref*{#1}}}
\newrefformat{hyp}{\savehyperref{#1}{Hypothesis~\ref*{#1}}}
\newrefformat{alg}{\savehyperref{#1}{Algorithm~\ref*{#1}}}
\newrefformat{rem}{\savehyperref{#1}{Remark~\ref*{#1}}}
\newrefformat{item}{\savehyperref{#1}{Item~\ref*{#1}}}
\newrefformat{step}{\savehyperref{#1}{step~\ref*{#1}}}
\newrefformat{conj}{\savehyperref{#1}{Conjecture~\ref*{#1}}}
\newrefformat{fact}{\savehyperref{#1}{Fact~\ref*{#1}}}
\newrefformat{prop}{\savehyperref{#1}{Proposition~\ref*{#1}}}
\newrefformat{prob}{\savehyperref{#1}{Problem~\ref*{#1}}}
\newrefformat{claim}{\savehyperref{#1}{Claim~\ref*{#1}}}
\newrefformat{relax}{\savehyperref{#1}{Relaxation~\ref*{#1}}}
\newrefformat{red}{\savehyperref{#1}{Reduction~\ref*{#1}}}
\newrefformat{part}{\savehyperref{#1}{Part~\ref*{#1}}}

\newcommand{\Sref}[1]{\hyperref[#1]{\S\ref*{#1}}}

\usepackage{nicefrac}

\let\nfrac=\nicefrac

\newcommand{\half}{\nicefrac12}

\ifnum\usemicrotype=1
\usepackage{microtype}
\fi

\ifnum\showauthornotes=1
\newcommand{\Authornote}[2]{{\sffamily\small\color{red}{[#1: #2]}}}
\newcommand{\Authorcomment}[2]{{\sffamily\small\color{gray}{[#1: #2]}}}
\newcommand{\Authorstartcomment}[1]{\sffamily\small\color{gray}[#1: }

\newcommand{\Authorfnote}[2]{\footnote{\color{red}{#1: #2}}}
\newcommand{\Authorfixme}[1]{\Authornote{#1}{\textbf{??}}}
\newcommand{\Authormarginmark}[1]{\marginpar{\textcolor{red}{\fbox{\Large #1:!}}}}
\else
\newcommand{\Authornote}[2]{}
\newcommand{\Authorcomment}[2]{}
\newcommand{\Authorstartcomment}[1]{}

\newcommand{\Authorfnote}[2]{}
\newcommand{\Authorfixme}[1]{}
\newcommand{\Authormarginmark}[1]{}
\fi

\ifnum\showfixme=0

\fi

\usepackage{boxedminipage}

\newcommand{\norm}[1]{\lVert#1\rVert}

\usepackage{dsfont}
\usepackage{mathrsfs}

\newcommand{\textparen}[1]{\text{(#1)}}

\ifx\because\undefined
\newcommand{\because}[1]{\textparen{because #1}}
\else
\renewcommand{\because}[1]{\textparen{because #1}}
\fi

\newcommand{\symdiff}{\Delta}

\newcommand{\mper}{\,.}

\newcommand\bdot\bullet

\DeclareMathOperator{\vol}{vol}

\newcommand{\N}{\mathbb N}
\newcommand{\R}{\mathbb R}

\renewcommand{\leq}{\leqslant}
\renewcommand{\le}{\leqslant}
\renewcommand{\geq}{\geqslant}
\renewcommand{\ge}{\geqslant}

\ifnum\showdraftbox=1

\else

\fi

\let\epsilon=\varepsilon

\numberwithin{equation}{section}

\newcommand{\MYstore}[2]{%
  \global\expandafter \def \csname MYMEMORY #1 \endcsname{#2}%
}

\newcommand{\MYload}[1]{%
  \csname MYMEMORY #1 \endcsname%
}

\newcommand{\MYnewlabel}[1]{%
  \newcommand\MYcurrentlabel{#1}%
  \MYoldlabel{#1}%
}

\newcommand{\MYdummylabel}[1]{}

\newcommand{\torestate}[1]{%
  \let\MYoldlabel\label%
  \let\label\MYnewlabel%
  #1%
  \MYstore{\MYcurrentlabel}{#1}%
  \let\label\MYoldlabel%
}

\newcommand{\restatetheorem}[1]{%
  \let\MYoldlabel\label
  \let\label\MYdummylabel
  \begin{theorem*}[Restatement of \prettyref{#1}]
    \MYload{#1}
  \end{theorem*}
  \let\label\MYoldlabel
}

\newcommand{\restatelemma}[1]{%
  \let\MYoldlabel\label
  \let\label\MYdummylabel
  \begin{lemma*}[Restatement of \prettyref{#1}]
    \MYload{#1}
  \end{lemma*}
  \let\label\MYoldlabel
}

\newcommand{\restateprop}[1]{%
  \let\MYoldlabel\label
  \let\label\MYdummylabel
  \begin{proposition*}[Restatement of \prettyref{#1}]
    \MYload{#1}
  \end{proposition*}
  \let\label\MYoldlabel
}

\newcommand{\restatefact}[1]{%
  \let\MYoldlabel\label
  \let\label\MYdummylabel
  \begin{fact*}[Restatement of \prettyref{#1}]
    \MYload{#1}
  \end{fact*}
  \let\label\MYoldlabel
}

\newcommand{\restate}[1]{%
  \let\MYoldlabel\label
  \let\label\MYdummylabel
  \MYload{#1}
  \let\label\MYoldlabel
}

\newcommand{\sse}{\subseteq}

\newcommand{\eps}{\epsilon}

\let\origparagraph\paragraph
\renewcommand{\paragraph}[1]{\origparagraph{#1.}}

\allowdisplaybreaks

\title{Gap Amplification for Small-Set Expansion via Random Walks}
\author{Prasad Raghavendra\thanks{University of California, Berkeley. Research supported by NSF Career Award and Alfred Sloan P. Fellowship. Email: \url{prasad@cs.berkeley.edu}}
and Tselil Schramm\thanks{University of California, Berkeley. This material is
based upon work supported by a Berkeley Chancellor's Fellowship and the National
Science Foundation Graduate Research Fellowship Program under Grant No. DGE
1106400. Email:
\url{tschramm@cs.berkeley.edu}}}
\date{} 

\begin{document}
\maketitle

\abstract{In this work, we achieve gap amplification for the Small-Set
	Expansion problem. Specifically, we show that an instance of
	the Small-Set Expansion Problem with completeness $\epsilon$
	and soundness $\frac{1}{2}$ is at least as difficult as
	Small-Set Expansion with completeness $\epsilon$ and soundness
	$f(\epsilon)$, for any function $f(\epsilon)$ which grows
	faster than $\sqrt{\epsilon}$. We achieve this amplification
	via random walks -- the output graph corresponds to taking
	random walks on the original graph. An interesting feature of our reduction is that unlike gap amplification via parallel repetition,  the size of the instances (number of vertices) produced by the reduction remains the same.
}
\newpage
\section{Introduction}

The small-set expansion problem refers to the problem of approximating
the edge expansion of small sets in a graph.  Formally, given a graph
$G = (V,E)$ and a subset of vertices $S \sse V$ with $|S| \leq |V|/2$, the edge expansion of
$S$ is
$$\phi(S) = \frac{E(S,\bar{S})}{\vol(S)},
$$
where $\vol(S)$
refers to the fraction of all edges of the graph that are incident on the subset $S$.
The edge expansion of the graph $G$ is given by
$ \phi_G = \min_{S \sse V, \vol(S) \leq \half} \phi(S)$.
The problem of approximating the value of $\phi_G$ is the
well-studied uniform sparsest cut problem \cite{LeightonR99,
AroraRV04, AroraLN05}.

In the small-set expansion problem, the goal is to approximate the
edge expansion of the graph at a much finer granularity.
Specifically, for $\delta > 0$ define the parameter
$\phi_G(\delta)$ as follows:
$$
\phi_{G}(\delta) =
\min_{S \sse V, \vol(S) \leq \delta } \phi(S).
$$
The problem of approximating $\phi_G(\delta)$ for all $\delta > 0$ is
the small-set expansion problem.

The small-set expansion problem has received considerable attention in
recent years due to its close connections to the unique games
conjecture.  To describe this connection, we will define a gap version
of the problem.
\begin{defin} \label{def:sse}
For constants $0 < s < c < 1$ and $\delta > 0$, the $SSE_\delta(c,s)$
problem is defined as follows: Given a graph $G = (V,E)$ distinguish
between the following two cases:
\begin{itemize}
  \item $G$ has a set $S$ with $\vol(S) \in [\delta/2,\delta]$ with expansion less than $1 - c$
\item All sets $S$ with $\vol(S) \leq \delta$ in $G$ have expansion at least $1 - s$.
\end{itemize}
\end{defin}
We will omit the subscript $\delta$ and write $SSE(c,s)$ when we refer
to the $SSE_{\delta}(c,s)$ problem for all constant $\delta > 0$.

Recent work by Raghavendra and Steurer \cite{RaghavendraS10}
introduced the following hardness assumption and showed that it
implies the unique games conjecture.
\begin{hypothesis}
	For all $\epsilon > 0$, there exists $\delta > 0$ such that
	$SSE_{\delta}(1-\epsilon,\epsilon)$ is $NP$-hard.
\end{hypothesis}

\begin{theorem} \cite{RaghavendraS10}
	The small set expansion hypothesis implies the unique games
	conjecture.
\end{theorem}

Moreover, the small set expansion hypothesis is shown to be equivalent
to a variant of the Unique Games Conjecture wherein the input instance
is promised to be a small-set expander \cite{RaghavendraST12}.  Assuming
the small-set expansion hypothesis, hardness results have been
obtained for several problems including Balanced Separator, Minimum
Linear Arrangment \cite{RaghavendraST12} and the problem of approximating vertex expansion
\cite{LouisRV13}.

In this work, we will be concerned with gap amplification for the
small set expansion problem.  Gap amplification
refers to an efficient reduction that takes a weak hardness result for a problem
$\Pi$ with a small gap between the completeness and soundness and
produces a strong hardness with a much larger gap.  Formally, this is achieved
via an efficient reduction from instances of problem $\Pi$ to {\it
harder} instances of the same problem $\Pi$.  Gap amplification is a crucial
step in proving hardness of approximation results.  An important
example of gap amplification is the parallel repetition of $2$-prover
1-round games or Label Cover.  Label cover is a constraint
satisfaction problem which is the starting point for a
large number of reductions in hardness of approximation
\cite{Hastad01}.  Starting with the PCP theorem, one obtains a weak
hardness for label cover with a gap of $1$ vs $1-\beta_0$ for some
tiny absolute constant $\beta_0$ \cite{AroraLMSS98}.  Almost all
label-cover based hardness results rely on the much stronger  $1$ vs
$\epsilon$ hardness for label cover obtained by gap amplification via
the parallel repetition theorem of Raz \cite{Raz98}.  More recently,
there have been significant improvements and simplifications to the
parallel repetition theorem \cite{Rao08,Holenstein07,DinurS13}.

It is unclear if parallel repetition could be used for gap
amplification for small set expansion.  Given a graph $G$, the
parallel repetition of $G$ would consist of the product graph $G^R$
for some large constant $R$.  Unfortunately, the product graph $G^R$
can have small non-expanding sets even if $G$ has no small
non-expanding sets.  For instance, if $G$ has a balanced cut
then $G^R$ could have a non-expanding set of volume $\frac{1}{2^R}$.

In this work, we show that random walks can be used to achieve gap
amplification for small set expansion.  Specifically, given a graph
$G$ the gap amplification procedure constructs $G^t$ on the same set
of vertices as $G$, but with edges corresponding to $t$-step lazy random
walks in $G$.  Using this approach, we are able to achieve the
following gap amplification.
\begin{thm} \label{thm:amplification}
Let $f$ be any function such that $\lim_{\epsilon \to 0} \frac{f(\epsilon)}{\sqrt{\epsilon}} \to \infty$. Then

If for all $\epsilon > 0$, $SSE'(1-\epsilon, 1-f(\epsilon))$ is NP-hard
then
for all $\eta > 0$, $SSE(1-\eta, \half)$ is NP-hard.
\end{thm}
We remark here that the result has some discrepancy in the set
sizes between the original instance and the instance produced by the
reduction.  For this reason, the reduction has to start with a slightly
different version of the Small set expansion problem $SSE'$ (See
\prettyref{def:sseprime}).

The above result nicely complements the gap
amplification result for the closely related problem of Unique Games
obtained via parallel repetition \cite{Rao08}.  For the sake of
completeness we state the result below.
\begin{theorem} \cite{Rao08}
Let $f$ be any function such that $\lim_{\epsilon \to 0} \frac{f(\epsilon)}{\sqrt{\epsilon}} \to \infty$. Then

If for all $\epsilon > 0$ if $UG(1-\epsilon, 1-f(\epsilon))$ is NP-hard
then
for all $\eta > 0$, $UG(1-\eta, \half)$ is NP-hard.
\end{theorem}
Note that the size of the instance produced by our reduction
remains bounded by $O(n^2)$.  In fact, the instance produced has the
same number of vertices but possibly many more edges.  This is in
contrast to parallel repetition wherein the size of the instance grows
exponentially in the number of repetitions used.

Technically, the proof of the result is very similar to an argument
in the work of Arora, Barak and Steurer
\cite{AroraBS10} to show that graphs with sufficiently high threshold rank
cannot be small-set expanders  (see Steurer's thesis
\cite{SteurerThesis10} for an improved
version of the result).  The work of O'Donnell and Wright \cite{ODonnellW} recast these
arguments using continous-time random walks instead of lazy-random
walks, yielding cleaner and more general proofs.
In this work, we will reuse the proof technique and obtain upper and
lower bounds for the expansion profile of lazy random walks
(see \prettyref{thm:expansionprofile}).  These
upper and lower bounds immediately imply the desired gap amplification
result for small-set expansion.

Subsequent to our work, Kwok and Lau \cite{KwokL14} have obtained a
stronger analysis of our gap amplification theorem, yielding almost
tight bounds.

\section{Preliminaries}

Unless otherwise specified, we will be concerned with an undirected
 graph $G = (V,E)$ with $n$ vertices and associated edge weights $w: E \to \R^+$.  The
degree of vertex $i$ denoted by $d(i) = \sum_{ (i,j) \in E} w(i,j)$.
The volume of a set $S\subseteq V$ is defined to be $\vol(S) = \sum_{i \in S} d(i)$.
Henceforth, we will assume that the total volume is $1$, i.e.,
$\sum_{i \in V} d(i) = 1$.
The adjacency matrix $A$ of the graph $G$ has entries $A_{ij} =
w(i,j)$.  The degree matrix $D$ is a $n \times n$ diagonal matrix with
$D_{ii} = d(i)$.

\paragraph{Expansion profile}  The expansion profile of a graph is
defined as follows.
\begin{definition}
  For a graph $G$, define the expansion profile $\phi_G : \R^+ \to
  [0,1]$ as
  $$ \phi_G(\delta) = \min_{S \subseteq V, \vol(S) \leq \delta}
  \phi(S)$$
  where $\phi(S) = \frac{E(S, \bar{S})}{\vol(S)}$.
\end{definition}

\paragraph{Lazy Random Walks}
The transition matrix for a lazy random walk on $G$
is given by
$$
M = \tfrac{1}{2}(I + D^{-1}A)
$$
The lazy random walk corresponds to staying at the same vertex with
probability $\frac{1}{2}$, and moving to a random neighbor with
probability $\frac{1}{2}$.
We will let $G^t$ denote the graph corresponding to the $t$-step lazy random
walk. The adjacency matrix of $G^t$ is given by $DM^t$.

We recall a few standard facts about  lazy random walks here.
\begin{fact}
\label{fact:squarednorm}
If $G$ is a graph with adjacency matrix $A$,  then $G$'s lazy random
walk operator $M = \frac{1}{2}(I + D^{-1}A)$ has the property that
$\norm{D^{\half}Mv}^2_2 = v^T D M^2 v$ for any vector $v$.
\end{fact}
\begin{proof}
	We use the fact that $M = \tfrac{1}{2} D^{-\half}(I + D^{-\half}AD^{-\half}))D^{\half}$:
\begin{align*}
\norm{D^{\half}Mv}^2_2
& = \tfrac{1}{4} v^T M^T D M v \\
& = \tfrac{1}{4}v^T D^{\half}(I + D^{-\half} A D^{-\half})D^{-\half} D D^{-\half}(I + D^{-\half} A D^{-\half}) D^{\half}v \\
& = \tfrac{1}{4}v^T D^{\half}(I + D^{-\half} A D^{-\half})^2D^{\half}v \\
& = v^T D M^2 v,
\end{align*}
as desired.
\end{proof}

\begin{fact}
\label{fact:2normIneq}
If $G$ is a graph with adjacency matrix $A$, then for the lazy random
walk operator $M = \frac{1}{2}(I + D^{-1}A)$, we have
\[
  \exNt{D^{\nicefrac{1}{2}}v}^2  = v^T D v \ge v^T DMv \ge v^TDM^2 v = \exNt{D^{\nicefrac{1}{2}}Mv}^2.
\]
\end{fact}
\begin{proof}
  Since the eigenvalues $\lambda_i$ of $D^{-\half}A D^{-\half}$ are
  between $[-1,1]$, the eigenvalues of
  $M' = \frac{1}{2}(I + D^{-\half}AD^{-\half})$ are
  $\mu_i = \frac{1}{2}(1 + \lambda_i)$, and so $\mu_i \in [0,1]$. Let
  $D^{\half}v = \sum \alpha_i  u_i$ be the decomposition of $D^{\half} v$
  in terms of the eigenvectors of $M'$. Then we have
  \[
	D^{\half}M v
	= M' D^{\half} v = \sum \alpha_i \mu_i u_i,
  \]
  and so
  $ v^TDv = \sum \alpha_i^2$, $v^TDMv = \sum \alpha_i^2 \mu_i$, and $v^TDM^2 v = \sum \alpha_i^2 \mu_i^2$.
  Since $\mu_i \in [0,1]$, we have $v^TDv \ge v^T D M v \ge v^T D M^2 v$, as desired.
\end{proof}

\begin{fact}
\label{fact:1normpres}
For the lazy random walk operator $M = \frac{1}{2}(I + D^{-1}A)$ and
any vector $v \in \R^V$, $v \ge 0$ we have
\[
\exNo{Dv} = \exNo{DMv}.
\]
\end{fact}
\begin{proof}
Let $v \in \R^V$. We have
\begin{align*}
\norm{DMv}_1
&= 1^TD(\tfrac{I + D^{-1}A}{2})v
~=~ \tfrac{1}{2}((1^TD)v + (1^TA)v)
~=~ \norm{Dv}_1,
\end{align*}
where the last inequality follows because $1^TD = 1^TA$.
\end{proof}

\paragraph{Small-Set Expansion Problem}
The formal statement of the SSE' problem is as follows.
\begin{defin} \label{def:sseprime}
For constants $0 < s < c < 1$ and $\delta > 0$, the Small-Set
Expansion problem $SSE'_\delta(c,s)$
 is defined as follows: Given a graph $G = (V,E)$, distinguish
between the following two cases:
\begin{itemize}
	\item $G$ contains a set $S$ such that $\vol(S) \in [\delta/2,
			\delta]$ and
	$\phi(S) \leq 1 - c$
\item All sets $S$ with $ \vol(S) \leq \mathbf{8\delta}$ in $G$ have expansion
	$\phi(S) \geq 1 - s$.
\end{itemize}
\end{defin}
The key difference from $SSE_\delta(c,s)$ is that the soundness is
slightly stronger in that even sets of size $8\delta$ have expansion
at least $1-s$.

\paragraph{Organization}
In \prettyref{sec:expprofile}, we will obtain upper and lower bounds
(\prettyref{thm:expansionprofile}) for expansion profile of lazy
random walks.  Subsequently, we use these bounds to conclude the main
result of the paper in \prettyref{sec:reduction}.
In \prettyref{app:equiv}, we give a reduction that establishes the
equivalence of the search versions of two different notions of Small-Set
Expansion.  We also present a reduction from $SSE$ on irregular graphs to $SSE$ on regular graphs in
\prettyref{app:irregToReg}.  Finally, in \prettyref{app:UG}, we discuss some obstacles encountered in applying
this reduction to the Unique Games problem.

\section{Expansion profile of lazy random walks } \label{sec:expprofile}

Let $G = (V,E)$ be a graph with adjacency matrix $A$, and
diagonal degree matrix $D$. The transition matrix for a lazy random walk on $G$
is $M = \tfrac{1}{2}(I + D^{-1}A)= \tfrac{1}{2}D^{-\half}(I +
D^{-\half}AD^{-\half})D^{\half}$.

For every $t \in \N$, let $G^t$ denote the graph corresponding to the $t$-step lazy random
walk whose adjacency matrix is given by $DM^t$.  We will prove the
following theorem about the expansion profile of $G^t$.

\begin{theorem} \label{thm:expansionprofile}
  For all $t \in \N$ and $\eta, \delta \in (0,1]$,
	if $G^t$ denotes the graph corresponding to the $t$-step lazy
	random walk in a graph $G = (V,E)$ then,
	$$ \min\left( 1 - \left(1-
  \frac{\phi^2_G(\frac{4\delta}{\eta})}{32}\right)^{t},
    1-\eta \right)   \leq \phi_{G^t}(\delta) \leq \frac{t}{2} \cdot \phi_G(\delta) $$
\end{theorem}
We will split the proof of the above theorem in to two parts:
\prettyref{lem:upperbound} and \prettyref{lem:lowerbound}

\begin{lemma} \label{lem:upperbound}
	For every subset $S \sse V$,
  $$ \phi_{G^t}(S) \leq \frac{t}{2} \cdot
  \phi_G(S),$$
and therefore
  $\phi_{G^t}(\delta) \leq \frac{t}{2} \cdot
  \phi_G(\delta)$.
\end{lemma}

\begin{proof}
  Fix a subset $S \subset V$.  From \cite{GharanT12}, we have that the probability $p(t)$ that a
lazy random walk stays entirely in $S$ for $t$ steps is bounded below by \[
p(t) \ge \left(1 - \frac{1}{2}\phi(S)\right)^t.  \] Now, the expansion of $S$
in $G^t$ is the probability of leaving the set on the $t$th step of the random
walk, which is at most $1 - p(t)$. Hence,
\[
  \phi_{G^t}(S)
\le 1 - p(t)
\le 1 - \left(1 - \frac{1}{2}\phi(S)\right)^t
\le \frac{t}{2}\phi(S),
\]
as desired. The result immediately follows for all sets of volume $\le \delta$.
\end{proof}

\begin{lemma}  \label{lem:lowerbound}

  For all $t, \eta$,
  $$ \phi_{G^t}\left(\delta\right) \geq \min\left( 1 - \left(1-
  \frac{\phi^2_G(\frac{4\delta}{\eta})}{32}\right)^{t},
    1-\eta \right) $$

 \end{lemma}

We prove this lemma by contradiction, by showing that if the expansion in the
final graph is not large enough then there exists a vector with bounded
Rayleigh quotient with respect to the original graph, from which we can extract
a non-expanding set.  The intuition is that the expansion of a set in the final
graph $DM^t$ corresponds to the neighborhood of the random walk after $t$
steps, and if the neighborhood is not large enough after $t$ steps, there must
be at least one step (or application of $M$) during which it did not grow.

\begin{proof} Suppose by way of contradiction that this is not the
  case.
  Let $\beta = \phi_G(\frac{4\delta}{\eta})$ and
  let $\delta' = \frac{4\delta}{\eta}$.  Further, let $\hat \beta =
  \frac{1}{2} \beta$.

Let $S$ be a set of volume at most $\delta \cdot \vol(V)$ such that
\begin{equation} \label{eq:contra}
\phi_{G^t}(S) \leq \min\left( 1 - \left(1-
\frac{\hat{\beta}^2}{8}\right)^{t},
    1-\eta \right).
 \end{equation}
  Let $v_0 = 1_S$ be the vector corresponding to
the indicator function of the set $S$. Define $v_{i} =  M^iv_{0}$, and for
the diagonal degree matrix $D$ of $A$, define
$w_i = D^{\half}v_i$. Note that $\exNt{w_0}^2 = \vol(S)$,
and $\norm{Dv_0}_1 = \vol(S)$.  By Fact~\ref{fact:1normpres} we also have
$\norm{Dv_i}_1 = \vol(S)$ for all $i$.

We first lower-bound $\exNt{w_{\frac{t}{2}}}$.  By definition of expansion,
\begin{align}
\phi_{G^t}(S)
&= 1 - \frac{v_0^T D M^t v_0}{v_0^T D v_0} \nonumber \\
\intertext{which by Fact~\ref{fact:squarednorm} implies that $\exNt{D^{\half}M^\frac{t}{2} v_0}^2 = \vol(S)
(1-\phi_{G^t}(S))$.  Now, using \eqref{eq:contra} we get}
\exNt{w_{\frac{t}{2}}}^2
~=~ \exNt{D^{\half} M^\frac{t}{2} v_0}^2
&=  \vol(S) (1- \phi_{G^t}(S))
~\geq~ \vol(S) \cdot \max\left(\eta, (1-\frac{1}{8}\hat\beta^2)^t\right)
	\label{eq:lbd1}
\end{align}
By Fact~\ref{fact:2normIneq}, we have $\exNt{w_i} \ge \exNt{w_{i+1}} \ge 0$
for all $i$, and (\ref{eq:lbd1}) holds for all $i \le \frac{t}{2}$.

We now assert that there must be some $i$ for which
\[
\qquad \frac{\exNt{w_{i+1}}^2}{\exNt{w_i}^2} >
1-\tfrac{1}{4}\hat\beta^2.
\]
To see this, consider the product of all such terms for $i < \frac{t}{2}$. Some algebraic simplification shows that
\[
	\prod_{i = 0}^{\frac{t}{2}-1}
	\frac{\norm{w_{i+1}}^2_2}{\norm{w_i}^2_2}
~=~ \frac{\exNt{w_{\frac{t}{2}}}^2}{\exNt{w_0}^2}
~>~ \frac{(1-\tfrac{1}{8}\hat\beta^2)^t \cdot \vol(S)}{\vol(S)}
~=~ \left(1-\tfrac{1}{8}\hat\beta^2\right)^t,
\]
where the second-to-last inequality follows from (\ref{eq:lbd1}).
Thus for some $i < \frac{t}{2}$ we have
\[
	\frac{\norm{w_{i+1}}^2_2}{ \norm{w_i}^2_2}
	~>~ \left((1-\frac{1}{8}\hat\beta^2)^t\right)^{\frac{2}{t}}
	~ > ~ 1-\frac{1}{4}\hat\beta^2.
\]

Then let $w_i$ be the vector corresponding to the first $i$ for which
$\exNt{w_{i+1}}^2 \ge (1-\frac{1}{4}\hat\beta^2)\exNt{w_i}^2$.

Since $w_{i+1}$ is obtained from $v_i$ via one step of a lazy random walk and a
normalization, we can bound the Rayleigh quotient of $v_i$ with respect to the
Laplacian of $DM = \tfrac{1}{2}(D + A)$:
\begin{align}
\frac{v_i^T  D(I - M) v_i}{v_i^T D v_i}
& = 1 - \frac{v_i^T DM v_i}{v_i^TDv_i}, \nonumber \\
\intertext{by Fact~\ref{fact:2normIneq},}
& \leq  1 - \frac{ v_i^T  D  M^2 v_i}{v_i^TDv_i} \qquad  \nonumber\\
\intertext{and by Fact~\ref{fact:squarednorm},}
& =  1 - \frac{ \norm{w_{i+1}}_2^2}{\norm{w_i}_2^2} \nonumber \\
& \le \frac{1}{4}\hat\beta^2.  \label{eq:ray}
\end{align}

We now truncate the vector $v_i$, then run Cheeger's algorithm on the truncated vector in order to find a non-expanding small set, and thus obtain a contradiction.
 Let $\theta = \frac{\eta}{4}$.
We take the truncated vector
\[
z_i(j) = \left\{\begin{array}{ll} v_i(j) - \theta &\ v_i(j) \ge \theta\\
0 &\ \text{otherwise}\end{array}\right..
\]
By Fact~\ref{fact:1normpres}, $Dv_i$ has $L_1$ mass $\vol(S)$.
Thus, the total volume of the set $S_z$ of vertices with nonzero support in $z_i$ is
\[
\vol(S_z)
~=~
\sum_{v_i(j) > \theta} d(j)
~\le~
\sum_{v_i(j) > \theta} \frac{1}{\theta} d(j)v_i(j)
~\le~ \frac{1}{\theta} \cdot \norm{Dv_i}_1 = \frac{4 \vol(S)}{\eta}
\]
Hence any subset of $S_z$ has volume at most $ \frac{4 \vol(S) }{\eta}$.

For the vector $v_i$, we know that $\norm{Dv_i}_1
= \vol(S)$.  Moreover using \eqref{eq:lbd1},
\begin{align*}
\norm{ D^{1/2} v_i}_2^2
= \norm{w_i}_2^2 \geq \norm{w_{t/2}}_2^2 \geq \eta \vol(S) \mper
\end{align*}
Applying Lemma~\ref{lem:rayleighbound} to $v_i$ and $z_i$ to conclude,
$$ \frac{z_i^TD(I - M)z_i}{z_i^TDz_i} \leq 2 \frac{v_i^TD(I -
M)v_i}{v_i^TDv_i} \mper$$
Using (\ref{eq:ray}), this implies the following bound on the Rayleigh
quotient of $z_i$,
\[
  \frac{z_i^TD(I - M)z_i}{z_i^TDz_i} \ \leq\  \frac{1}{2}\hat\beta^2
  \mper
\]

Thus, when we run Cheeger's algorithm on $z_i$, we get a set of volume at most
$\frac{4\vol(S)}{\eta}$ and of expansion less than $\hat\beta$ in $DM$, and therefore less
than $\beta$ in $G$.  Since $\beta =
\phi_G(\frac{4\delta}{\eta})$, this is a contradiction.  This completes the proof of
Lemma~\ref{lem:lowerbound}.
\end{proof}

The following lemma, which gives an upper bound on the Rayleigh quotient of a
truncated vector, is a slight generalization of Lemma 3.4 of \cite{AroraBS10}.

\begin{lem}
\label{lem:rayleighbound}
Let $x \in \R^V$ be non-negative, let $L$ be the weighted Laplacian of a graph
$G = (V,E)$ with weights $w(i,j)$ and degree matrix $D$.  Suppose that
\begin{equation}
\label{eq:sparsitycond}
4\theta \norm{D x}_1 \le \norm{D^{1/2} x}_2^2
\end{equation}
Then for the threshold vector $y$ defined by
\[
y(i) = \left\{ \begin{array}{ll} x(i) -\theta & \ x(i) > \theta \\ 0 & \ \text{otherwise}\end{array}\right.,
\]
we have
\[
\frac{y^TLy}{y^TDy} ~\le~ 2 \cdot \frac{ x^TLx}{x^TDx}.
\]
\end{lem}

\begin{proof}
First, we show $y^TLy \le x^TLx$.
\begin{eqnarray*}
y^TLy &=& \sum_{(i,j) \in E} w(i,j)(y(i) - y(j))^2\\
&=& \sum_{\substack{(i,j) \in E\\ y(i), y(j) \ge 0}} w(i,j)(x(i) - x(j))^2 + \sum_{\substack{(i,j) \in E\\ y(i) \ge 0,  y(j) = 0}} w(i,j)(x(i) - \theta)^2\\
&\le& \sum_{(i,j) \in E} w(i,j)(x(i) - x(j))^2\\
&=& x^TLx,
\end{eqnarray*}
where the second-to-last inequality follows from the fact that if $y(i) = 0$, then $x(i) \le \theta$.

Now, we show that $y^TDy \ge \frac{1}{2}x^TDx$. First, we note that $d(i)y(i)^2 \ge d(i)x(i)^2 - 2\theta d(i)x(i)$ for all $k$. Thus,
\begin{eqnarray*}
\sum_{i \in V}d(i)y(i)^2 &\ge&  \sum_{i \in V}d(i)x(i)^2 - 2\theta d(i)x(i)\\
&=& \left( \sum_{i \in V}d(i)x(i)^2\right) - 2\theta \left(\sum_{i \in V} d(i)x(i)\right) \\
&\ge& \frac{1}{2}\sum_{i \in V}d(i)x(i)^2
\end{eqnarray*}
Where the the last inequality follows by assumption \prettyref{eq:sparsitycond}.

Thus, we have
\[
\frac{y^TLy}{y^TDy}~\le~2\cdot \frac{x^TLx}{x^TDx} ,
\]
as desired.
\end{proof}

\section{Gap Amplification } \label{sec:reduction}
In this section, we will prove \prettyref{thm:amplification} which we restate here for convenience.
\begin{thm} (Restatement of \prettyref{thm:amplification})
Let $f$ be any function such that $\lim_{\epsilon \to 0} \frac{f(\epsilon)}{\sqrt{\epsilon}} \to \infty$. Then

If for all $\epsilon > 0$, $SSE'(1-\epsilon, 1-f(\epsilon))$ is NP-hard
then
for all $\eta > 0$ $SSE(1-\eta, \frac{1}{2})$ is NP-hard.
\end{thm}
\begin{proof}
	Fix $\epsilon$ small enough so that
	$\frac{64\epsilon}{f(\epsilon)^2} \leq \eta$.  There exists such
	an $\epsilon$ since $\lim_{\epsilon \to 0} \frac{f(\epsilon)}
	{\sqrt{\epsilon}} \to \infty$.   Fix  $t = \frac{64}{f(\epsilon)^2}$.

	  Given an instance $G$ of $SSE'(1-\eps,1-f(\epsilon))$, the reduction just outputs the graph $G^t$
  obtained via $t$-step lazy random walks on $G$.  Since the adjacency matrix of $G'$ can be calculated with $\log t$ matrix multiplications, this reduction clearly runs in time $O(n^3\log t)$.

  \paragraph{Completeness}  If there exists a set of $S$ with $\vol(S)
  \in [\delta/2,\delta]$ and $\phi_G(S) \leq \epsilon$ then by \prettyref{lem:upperbound} the same
  set $S$ satisfies,
  $$ \phi_{G^t}(S) \leq \frac{t}{2} \phi_G(S)
   = \Theta\left(\frac{\epsilon}{f(\epsilon)^2}\right) \leq \eta$$

   \paragraph{Soundness}  If $\phi_G(8\delta) \geq f(\epsilon)$ then by
   applying \prettyref{lem:lowerbound}
   $$ \phi_{G^t}\left(\delta\right) \geq \min\left( 1 - \left(1-
   \frac{1}{32}f(\epsilon)^2\right)^t , \half\right) \geq \frac{1}{2}$$

\end{proof}

\bibliographystyle{amsalpha}
\bibliography{papers}

\newcommand{\etalchar}[1]{$^{#1}$}
\providecommand{\bysame}{\leavevmode\hbox to3em{\hrulefill}\thinspace}
\providecommand{\MR}{\relax\ifhmode\unskip\space\fi MR }
\providecommand{\MRhref}[2]{%
  \href{http://www.ams.org/mathscinet-getitem?mr=#1}{#2}
}
\providecommand{\href}[2]{#2}
\begin{thebibliography}{ALM{\etalchar{+}}98}

\bibitem[ABS10]{AroraBS10}
Sanjeev Arora, Boaz Barak, and David Steurer, \emph{Subexponential algorithms
  for unique games and related problems}, FOCS, 2010, pp.~563--572.

\bibitem[ALM{\etalchar{+}}98]{AroraLMSS98}
Sanjeev Arora, Carsten Lund, Rajeev Motwani, Madhu Sudan, and Mario Szegedy,
  \emph{Proof verification and the hardness of approximation problems}, JACM:
  Journal of the ACM \textbf{45} (1998).

\bibitem[ALN08]{AroraLN05}
Sanjeev Arora, James~R. Lee, and Assaf Naor, \emph{Euclidean distortion and the
  sparsest cut}, 1--21.

\bibitem[ARV04]{AroraRV04}
Sanjeev Arora, Satish Rao, and Umesh Vazirani, \emph{Expander flows, geometric
  embeddings and graph partitioning}, Proceedings of the thirty-sixth annual
  {ACM} Symposium on Theory of Computing ({STOC}-04) (New York), ACM Press,
  June ~13--15 2004, pp.~222--231.

\bibitem[DS13]{DinurS13}
Irit Dinur and David Steurer, \emph{Analytical approach to parallel
  repetition}, CoRR \textbf{abs/1305.1979} (2013).

\bibitem[GT12]{GharanT12}
Shayan~Oveis Gharan and Luca Trevisan, \emph{Approximating the expansion
  profile and almost optimal local graph clustering}, FOCS, 2012, pp.~187--196.

\bibitem[H\.01]{Hastad01}
Johann H\.astad, \emph{Some optimal inapproximability results}, Journal of the
  ACM \textbf{48} (2001), no.~4, 798--859.

\bibitem[Hol07]{Holenstein07}
Holenstein, \emph{Parallel repetition: Simplifications and the no-signaling
  case}, STOC: ACM Symposium on Theory of Computing (STOC), 2007.

\bibitem[KL14]{KwokL14}
Tsz~Chiu Kwok and Lap~Chi Lau, \emph{Personal communication}.

\bibitem[LR99]{LeightonR99}
Frank~Thomson Leighton and Satish Rao, \emph{Multicommodity max-flow min-cut
  theorems and their use in designing approximation algorithms}, J. ACM
  \textbf{46} (1999), no.~6, 787--832.

\bibitem[LRV13]{LouisRV13}
Anand Louis, Prasad Raghavendra, and Santosh Vempala, \emph{The complexity of
  approximating vertex expansion}, CoRR \textbf{abs/1304.3139} (2013).

\bibitem[OW12]{ODonnellW}
Ryan O'Donnell and David Witmer, \emph{Markov chain methods for small-set
  expansion}, Arxiv arXiv:1204.4688 (2012).

\bibitem[Rao08]{Rao08}
Anup Rao, \emph{Parallel repetition in projection games and a concentration
  bound}, STOC (Richard~E. Ladner and Cynthia Dwork, eds.), ACM, 2008,
  pp.~1--10.

\bibitem[Raz98]{Raz98}
Ran Raz, \emph{A parallel repetition theorem}, SIAM Journal on Computing
  \textbf{27} (1998), no.~3, 763--803.

\bibitem[RS10]{RaghavendraS10}
Prasad Raghavendra and David Steurer, \emph{Graph expansion and the unique
  games conjecture}, STOC, 2010, pp.~755--764.

\bibitem[RST12]{RaghavendraST12}
Prasad Raghavendra, David Steurer, and Madhur Tulsiani, \emph{Reductions
  between expansion problems}, IEEE Conference on Computational Complexity,
  2012, pp.~64--73.

\bibitem[Ste10]{SteurerThesis10}
David Steurer, \emph{On the complexity of unique games and graph expansion.},
  Ph.D. thesis, Princeton University, 2010.

\end{thebibliography}

\appendix

\section{Equivalence of Two Notions of the Small-Set Expansion Problem}
\label{app:equiv}

There is a slightly different version of the Small-Set expansion
decision problem that differs from \prettyref{def:sse} in the
soundness case.

\begin{defin}
\label{def:sse2}
For constants $0 < s < c < 1$, and $\delta > 0$,  the Small-Set
Expansion problem $SSE^{=}_\delta(c,s)$
 is defined as follows: Given a graph $G = (V,E)$  with $\vol(V) = N$, distinguish
between the following two cases:
\begin{itemize}
\item $G$ has a set of volume in the range $[\frac{1}{2} \delta N,  \delta N]$ with expansion less than $1 - c$
\item All sets in $G$ of volume in the range $[\frac{1}{4} \delta N, \delta N]$ have expansion at least $1 - s$.
\end{itemize}
\end{defin}

Clearly $SSE^{=}_{\delta}(c,s)$ is a harder decision problem than
$SSE_{\delta}(c,s)$ since the soundness assumption is weaker.
There is no known reduction from $SSE^{=}_{\delta}(c,s)$ to
$SSE_{\delta}(c,s)$ that establishes the equivalence of the two
versions. Here we observe that the search versions of these two
problems are equivalent.

\begin{prop}
For all $\delta_0,c,s > 0$
	A search algorithm for $SSE_{\delta}(2c-1,s)$ for  $ \delta \in [\delta_0/2,
	\delta_0]$  gives a search algorithm
for $SSE^{=}_{\delta}(c,s)$ in the range $\delta \in
[\delta_0/2,\delta_0]$.
\end{prop}
\begin{proof}
Suppose we are given an algorithm $A$ that finds a set $S'$ of volume at
most $\delta N $ and expansion less than $1-s$ whenever there exists a
set $S$ with $\vol(S) \in [\frac{1}{4}\delta N, \delta N]$ and $\Phi(S) \leq 2-2c$.
We construct a set $S \subseteq V$ such that
$\vol(S) \in [\frac{1}{4}
\delta N, \delta N ]$ and $\phi(S) < 1-s$. We proceed iteratively, as follows.

We start with an empy initial set, $S_{out}$, and with the full graph,
$G_0 = G$. If $\vol(S_{out}) \in  [\frac{1}{4} \delta N, \delta N]$, we
terminate and return $S_{out}$. Otherwise, at the $i$th step, we apply
$A$ to $G_{i-1}$ to obtain a set $S_i$ of expansion less than $1-s$.
If $\vol(S_i) \in [\frac{1}{4}\delta N, \delta N]$ return $S_i$,
otherwise add the vertices in $S_i$ to $S_{out}$. We then set $G_i = G_{i - 1} \setminus S_i$. If no such set can be found, then we terminate and return no.

Clearly, this algorithm terminates and runs in polynomial time.  Suppose $S'$ is a nonexpanding set
with $\vol(S') \in [\frac{1}{2} \delta N, \delta N]$. As long as $S_{out}$
has volume smaller than $\frac{1}{4} \delta N$, $S'-S_{out}$ will have volume
at least $\vol(S')/2$ and has expansion at most $2\phi(S') \leq 2 - 2c$.
Hence by the assumption about algorithm $A$, it will return a set $S_i$ of
expansion at most $1-s$.
The check of the volume of $S_i$ ensures that $S_{out}$ will never go from below
the allowable volume range to above in a single step.
Finally if $S_i$ was never returned for any step $i$,
the union of all the sets $S_i$ has
expansion at most $1-s$ and volume in the range $[\delta N/4,\delta N]$.

\end{proof}

\section{Reduction from Irregular Graphs to Regular Graphs}
\label{app:irregToReg}
In this section, we present a reduction from small set expansion on irregular graphs to small set expansion on regular graphs. Specifically, we prove the following theorem.
\begin{thm}
\label{thm:irregToReg}
There exists  an absolute constant $C$ such that for all
$\gamma, \beta \in (0,1)$ there is a polynomial time reduction from
$SSE_\delta(1- \gamma, 1 -
\beta)$ on a irregular graph $G = (V,E)$ to $SSE_{\delta}(1 -
\gamma, 1-\nfrac{\beta}{C})$ on a $4$-regular graph $G' = (V', E')$
\end{thm}

\begin{proof}
The reduction is as follows: we replace each vertex $v \in V$ with a
$3$-regular expander $A_v$ on $\text{deg}(v)$ vertices. Using standard
constructions of 3-regular expanders, we can assume that the graphs
$A_v$ have edge expansion at least $\kappa = 0.01$.  Now, for each edge $(v,w) \in E$, we add an edge between a particular vertex in $A_v$ and $A_w$. The resulting graph on the expanders is $G'$, with $V' = \cup_{v \in V} A_v$. Note that $G'$ is $d$-regular, and that $|V'| = \sum_{v \in V} \text{deg}(v) = \vol(V)$, as desired.

For the completeness, we note that if a set $S \subset V$ with volume
at most $\delta |V|$ has $\phi_G(S) < \gamma$, then the set $S' =
\cup_{v \in S} A_v$ has the same number of edges leaving the set as
$S$, and the number of vertices in the set is equal to $\vol(S)$.
Thus, $\phi_{G'}(S') < \gamma/4$, as desired.

For soundness, suppose there is a set $S' \subset V'$ with $|S'| \le
\delta |V'|$ and $\phi_{G'}(S') < \beta$. Then we can partition $S'$
into sets corresponding to each $A_v$; let $B_v = S' \cap A_v$. Then consider the set
\[
S^* = \cup_{|B_v| \ge \frac{1}{2}|A_v|} A_v,
\]
the set of $A_v$ that overlap with $S'$ by at least half.
We will argue that $S^*$ has expansion at most $\frac{10}{\kappa}\beta$ in $G'$. First, by definition of expansion we have
\[
\beta \ge \phi_{G'}(S') = \frac{\sum_{v} E(B_v, \bar{S'})}{ 4 \sum_{v
\in V} |B_v|} =  \frac{\sum_{v} E(B_v, A_v \setminus B_v) + E(B_v,
\bar{S'}\setminus A_v)}{ 4 \sum_{v\in V} |B_v|},
\]
where we distinguish between boundary edges of $S'$ inside and outside of the $A_v$. In particular, we have
\[
4 \beta \sum_{v \in V} |B_v| \ge \sum_{v \in V} E(B_v, A_v \setminus B_v) .
\]
Now, we bound from below the number of boundary edges within $A_v$. Since $A_v$ is an expander with expansion $\kappa$, we have
\[
 E(B_v,\ A_v \setminus B_v) \ge \kappa \cdot \min(|B_v|, |A_v \setminus B_v|).
\]
Hence we will have,
$$ S' \symdiff S^* = \sum_{v} \min(|B_v|, |A_v \setminus B_v|) \leq
\frac{1}{\kappa} \sum_v E(B_v, A_v\setminus B_v) \leq
\frac{4\beta}{\kappa} \sum_{v \in V} |B_v| = \frac{4\beta}{\kappa} |S'|$$
Since $G'$ is a $4$-regular graph, we can upper bound the expansion of
 $S^*$ by
\[
	\phi_{G'}(S^*) \le  \frac{E[S',\bar{S'}] + 4 |S' \symdiff
S^*|}{ 4 |S'| -
	4 |S' \symdiff S^*|} \le \frac{4 \beta |S'| + 16
		\nfrac{\beta}{\kappa} |S'|}{ 4 |S'| -
			16 \nfrac{\beta}{\kappa} |S'|}
	\leq
	\frac{ \beta \left(1 +
		\nfrac{4}{\kappa}\right)}{ 1- \nfrac{4\beta }{\kappa} }
		\]
Thus, in $G$ the set $S = \{v \ | \ A_v \in S^*\}$ has expansion at
most $ \frac{10}{\kappa} \beta$, and $\vol(S) \in [\frac{1}{2}\delta\vol(V), 2 \delta\vol(V)]$, as desired.
\end{proof}

\section{Discussion of Unique Games}
\label{app:UG}

The following simple example illustrates why a similar reduction will
not work for Unique Games. Consider a unique games instance with two
disconnected components, one of size $s\cdot n $ in which only
$\frac{1}{q}$ of the constraints are satisfiable, and another component
of size $(1-s)n$ in which all constraints are perfectly satisfiable.
In this case, running a random walk on the label-extended graph
will not alter the number of satisfiable constraints.

A slight modification of this approach would at first seem to be a promising
avenue for overcoming this example at first: reweight the graph of
constraints by $(1-w)$, and add to it an expander of weight $w$
with arbitrary constraints. Now, in the previous example, following the
proof of our completess case we observe that after $t$ steps
of the random walk, the formerly perfectly satisfiable component has
expansion at most $tw$, the soundness is at most $1 - t(w + rs)$.
Similarly, in the completeness case, our reduction goes from
completeness $1-c$ to soundness $1-t(c+w)$. Thus, because the added constraints
contribute equally to the soundness of the bad example and the completeness
in general, this approach does not overcome the problem
of isolated components.

\end{document}